\documentclass[copyright,creativecommons]{eptcs}
\providecommand{\event}{SAIRP 2013}
\usepackage{amssymb}
\usepackage{amsthm}

\renewcommand{\ni}{\noindent}

\newcommand{\bc}{\begin{center}}
\newcommand{\ec}{\end{center}}
\newcommand{\be}{\begin{enumerate}}
\newcommand{\ee}{\end{enumerate}}
\newcommand{\bi}{\begin{itemize}}
\newcommand{\ei}{\end{itemize}}
\newcommand{\ba}{\begin{array}}
\newcommand{\ea}{\end{array}}

\newcommand{\by}[1]{\stackrel{ #1 }{\longrightarrow}}

\newcommand{\eset}[1]{ \{ #1 \}}
\newcommand{\dia}[1]{\langle #1 \rangle}

\newcommand{\V}{\mathsf{V}}
\newcommand{\pff}{{\tt ff}}
\newcommand{\por}{\vee}
\newcommand{\pand}{\wedge}

\newcommand{\prtt}{{\tt tt}}
\newcommand{\pnot}{\neg}

\newcommand{\den}[1]{\|#1\|}
\newcommand{\dentv}[1]{\den{#1}_{\V}^{\T}}

\newtheorem{prop}{Proposition}
\newtheorem{defin}{Definition}
\newtheorem{theorem}[prop]{Theorem}

\newtheorem{fact}[prop]{Fact}

\def\Var{\mathrm{Var}}
\def\Prop{\mathrm{Prop}} 
\def\w{\emph}
\def\Act{\mathrm{Act}}
\def\T{\mathsf{T}}
\def\S{\mathsf{S}}
\def\rstrut{\vrule height 2ex depth .75ex width 0pt}
\def\proofrule#1#2{{\rstrut #1 \over \rstrut #2}} 
\makeatletter
\def\@eqnnum{}
\def\@yeqncr{\@testopt\@xeqncr\@eqnskip}
\def\@eqnskip{1\jot}
\makeatother
\makeatletter
\def\newspacing #1{\def\baselinestretch{#1}\ifx\@currsize\normalsize
                   \@normalsize \else \@currsize\fi}
\makeatother

\newcount\ProofTreebias \ProofTreebias=50
\newdimen\ProofTreeovershoot \ProofTreeovershoot=0.1em 
\newdimen\ProofTreeannotspace \ProofTreeannotspace=0pt
\def\ProofTreestrut{\vrule depth 5.5pt height 12.5pt width 0pt }
\def\ProofTreerulestrut{\vrule depth 0pt height 0pt width 0pt }
\newbox\ProofTree\newskip\ProofTreespace\ProofTreespace=1em
\newdimen\ProofTreea\newdimen\ProofTreeb\newdimen\ProofTreec
\def\ProofTreedopremise#1{#1}
\def\PTirule{\aPTirule{\hskip-\ProofTreeannotspace}}
\def\PTiirule{\aPTiirule{\hskip-\ProofTreeannotspace}}
\def\PTiiirule{\aPTiiirule{\hskip-\ProofTreeannotspace}}
\def\gobblebrace#1{%
\def\gobblearg{#1}%
\afterassignment\agobblebrace\let\gobblething= }

\def\agobblebrace{%
\ifx\gobblething\ProofTreespace 
\let\next=\setskip
\else
\ifcat\bgroup\noexpand\gobblething
\let\next=\gobblearg
\else
\ifcat\space\noexpand\gobblething
\let\next=\bgobblebrace
\else
\def\next{\errmessage
{I was expecting optional spaces followed by a left
brace!}}%
\fi\fi\fi\next}
\def\bgobblebrace{\afterassignment\agobblebrace\let\gobblething=}
\def\setskip{\afterassignment\bgobblebrace\skip255}

\def\doannot#1{\setbox0=\vbox to 0pt{\vss\hbox{#1}\vss}}
\def\aPTiirule#1#2{%
\begingroup
\doannot{#1}
\setbox7=\hbox{\ProofTreestrut\ProofTreedopremise{#2}}\setbox8=\copy7
\gobblebrace{\bgroup\aftergroup\aPTiirulepartb}}

\def\aPTiirulepartb{%
\setbox1=\box\ProofTree
\dimen1=\ProofTreea \dimen2=\ProofTreeb \dimen3=\ProofTreec
\skip255=\ProofTreespace 
\gobblebrace{\bgroup\aftergroup\aPTiirulepartc}}

\def\aPTiirulepartc{%
\ProofTreespace=\skip255
\setbox2=\box\ProofTree
\dimen4=\ProofTreea \dimen5=\ProofTreeb \dimen6=\ProofTreec
\dimen0=\wd1
\advance\dimen0 -\dimen1
\advance\dimen0 -\dimen2
\advance\dimen0 \ProofTreespace
\advance\dimen0 \dimen4
\advance\dimen0 \dimen5
\advance\dimen0 \dimen6
\ifnum\wd7 <\dimen0
\dimen9=\dimen1
\advance\dimen9 \dimen2
\global\ProofTreea=\dimen9
\ProofTreeb=\dimen0
\advance\ProofTreeb -\wd7
\divide\ProofTreeb 100
\multiply\ProofTreeb \ProofTreebias
\global\ProofTreeb=\ProofTreeb
\setbox7=\hbox to\dimen0{\hskip\ProofTreeb\box7\hfil}
\else 
\advance\dimen0 -\wd7
\divide\dimen0 100 \multiply\dimen0 \ProofTreebias
\dimen9=\dimen1
\advance\dimen9 \dimen2
\advance\dimen9 \dimen0
\global\ProofTreea=\dimen9
\dimen0=\wd7
\global\ProofTreeb=0pt
\fi
\ifnum\dimen9 <0
\dimen8=-\dimen9
\global\ProofTreea=0pt
\else
\dimen8=0pt
\fi
\global\ProofTreec=\wd8
\global\setbox\ProofTree=\vtop{\offinterlineskip\halign
{\hskip\dimen8 ##\cr
\hskip\dimen9\box7\hfil\cr
\ProofTreerulestrut\hskip\dimen9\hskip-\ProofTreeovershoot\advance\dimen0 2\ProofTreeovershoot
\vrule width\dimen0 depth0pt height 0.4pt\hskip\ProofTreeannotspace\box0\hskip-\ProofTreeovershoot
\hfil\cr
\box1\hskip\ProofTreespace\box2\hfil\cr}}%
\endgroup}%

\def\PTaxiom#1{\global\setbox\ProofTree=\vtop{\hbox{\ProofTreestrut\ProofTreedopremise{#1}}}%
\global\ProofTreea=0pt\global\ProofTreeb=0pt\global\ProofTreec=\wd\ProofTree
}

\def\aPTirule#1#2{
\begingroup
\doannot{#1}
\setbox7=\hbox{\ProofTreestrut\ProofTreedopremise{#2}}\setbox8=\copy7
\gobblebrace{\bgroup\aftergroup\aPTirulepartb}}
\def\aPTirulepartb{%
\setbox1=\box\ProofTree
\dimen1=\ProofTreea \dimen2=\ProofTreeb \dimen3=\ProofTreec
\dimen0=\dimen3 
\ifnum\wd7 <\dimen0
\dimen9=\dimen1
\advance\dimen9 \dimen2
\global\ProofTreea=\dimen9
\ProofTreeb=\dimen0
\advance\ProofTreeb -\wd7
\divide\ProofTreeb 100
\multiply\ProofTreeb \ProofTreebias
\global\ProofTreeb=\ProofTreeb
\setbox7=\hbox to\dimen0{\hskip\ProofTreeb\box7\hfil}
\else
\advance\dimen0 -\wd7
\divide\dimen0 100 \multiply\dimen0 \ProofTreebias
\dimen9=\dimen1
\advance\dimen9 \dimen2
\advance\dimen9 \dimen0
\global\ProofTreea=\dimen9
\dimen0=\wd7
\global\ProofTreeb=0pt
\fi
\ifnum\dimen9 <0
\dimen8=-\dimen9
\global\ProofTreea=0pt
\else
\dimen8=0pt
\fi
\global\ProofTreec=\wd8 
\global\setbox\ProofTree=\vtop{\offinterlineskip\halign
{\hskip\dimen8 ##\cr
\hskip\dimen9\box7\hfil\cr
\ProofTreerulestrut\hskip\dimen9\hskip-\ProofTreeovershoot\advance\dimen0 2\ProofTreeovershoot
\vrule width\dimen0 depth0pt height 0.4pt \hskip\ProofTreeannotspace\box0\hskip-\ProofTreeovershoot
\hfil\cr\box1\hfil\cr}}%
\endgroup}%

\def\aPTiiirule#1#2{%
\begingroup
\doannot{#1}
\setbox7=\hbox{\ProofTreestrut\ProofTreedopremise{#2}}\setbox8=\copy7
\gobblebrace{\bgroup\aftergroup\aPTiiirulepartb}}

\def\aPTiiirulepartb{%
\setbox1=\box\ProofTree
\dimen1=\ProofTreea \dimen2=\ProofTreeb \dimen3=\ProofTreec
\gobblebrace{\bgroup\aftergroup\aPTiiirulepartc}}

\def\aPTiiirulepartc{%
\setbox3=\box\ProofTree
\gobblebrace{\bgroup\aftergroup\aPTiiirulepartd}}

\def\aPTiiirulepartd{%
\setbox2=\box\ProofTree
\dimen4=\ProofTreea \dimen5=\ProofTreeb \dimen6=\ProofTreec
\dimen0=\wd1
\advance\dimen0 -\dimen1
\advance\dimen0 -\dimen2
\advance\dimen0 \ProofTreespace
\advance\dimen0 \wd3
\advance\dimen0 \ProofTreespace
\advance\dimen0 \dimen4
\advance\dimen0 \dimen5
\advance\dimen0 \dimen6
\ifnum\wd7 <\dimen0
\dimen9=\dimen1
\advance\dimen9 \dimen2
\global\ProofTreea=\dimen9
\ProofTreeb=\dimen0
\advance\ProofTreeb -\wd7
\divide\ProofTreeb 100
\multiply\ProofTreeb \ProofTreebias
\global\ProofTreeb=\ProofTreeb
\setbox7=\hbox to\dimen0{\hskip\ProofTreeb\box7\hfil}%
\else 
\advance\dimen0 -\wd7
\divide\dimen0 100 \multiply\dimen0 \ProofTreebias
\dimen9=\dimen1
\advance\dimen9 \dimen2
\advance\dimen9 \dimen0
\global\ProofTreea=\dimen9
\dimen0=\wd7
\global\ProofTreeb=0pt
\fi
\ifnum\dimen9 <0
\dimen8=-\dimen9
\global\ProofTreea=0pt
\else
\dimen8=0pt
\fi
\global\ProofTreec=\wd8
\global\setbox\ProofTree=\vtop{\offinterlineskip\halign
{\hskip\dimen8 ##\cr
\hskip\dimen9\copy7\hfil\cr
\ProofTreerulestrut\hskip\dimen9\hskip-\ProofTreeovershoot\advance\dimen0 2\ProofTreeovershoot
\vrule width\dimen0 depth0pt height 0.4pt \hskip\ProofTreeannotspace\box0\hskip-\ProofTreeovershoot
\hfil\cr
\box1\hskip\ProofTreespace\box3\hskip\ProofTreespace\box2\hfil\cr}}%
\endgroup}%

\def\tree{%
\gobblebrace{\bgroup\aftergroup\treepartb}}
\def\treepartb{%
\box\ProofTree
}

\title{A  Proof System with Names for Modal Mu-calculus\footnote{To Dave who 
I first met in 1982 when we shared 
an office in Edinburgh where I learnt about denotational
semantics, least fixpoints and Edinburgh pubs.}}
\author{
    Colin Stirling
    \institute{School of Informatics\\ University of Edinburgh}\\
    \email{cps@inf.ed.ac.uk}}
\def\titlerunning{A proof system with names for modal mu-calculus}
\def\authorrunning{Colin Stirling}
\begin{document}

\maketitle

\begin{abstract}
Fixpoints are an important ingredient in semantics, abstract interpretation
and program logics. Their addition to a logic
can add considerable expressive power.
One general  issue is how to define 
proof systems for such logics. Here we examine proof systems for
modal logic with fixpoints \cite{Koz83}. We  present a  tableau proof system
for checking validity of formulas 
which uses names to keep track of unfoldings of fixpoint variables 
as devised in  \cite{Ju09}.
\end{abstract}

\section{Introduction}
Fixpoints are an important ingredient in semantics, abstract interpretation
and program logics. Their addition to a logic
can add considerable expressive power.
One general  issue is how to define 
proof systems for such logics. 
In this paper we consider modal mu-calculus, modal logic with fipoints,
see \cite{BS07} for a survey. 
Dave Schmidt  has used this logic to understand data flow analyis
\cite{Sch}.
Here our interest is more with developing \w{proof systems} for the logic. 

In this paper we describe a tableau proof system which checks when
a modal mu-calculus formula is valid. The system 
uses names to keep track of unfoldings of fixpoint variables. 
This idea originated in \cite{StW9} in the context of  model checking. 
For satisfiability
checking it  was used in \cite{LS01} for LTL and CTL and then for
modal mu-calculus  in \cite{Ju09}. 

In Section~\ref{secmodal} we describe the syntax and semantics of 
modal mu-calculus and in Section~\ref{secproof} we briefly examine 
approaches to  devising proof systems for this logic. The tableau proof
system based on names for checking  valid formulas is then presented in
Section~\ref{secnames} and shown to be both sound and complete.

\section{Modal Mu-calculus}
\label{secmodal}
Let $\Var$ be an (infinite) set of \w{variable names}, typically
indicated by $Z, Y, \ldots$; let $\Prop$ be a set of \w{atomic
propositions}, typically indicated by $P,Q,\ldots$; and let $\Act$ be a set
of \w{actions}, typically indicated by $a,b,\ldots$. The set of modal
mu-calculus formulas $\mu M$ (with respect to $\Var, \Prop, \Act$) is as 
follows.

\[ 
\phi \; ::= \; Z \;\mid\; P \;\mid\; \pnot\phi \;\mid\;
\phi_{1} \pand  \phi_{2} \;\mid\; [a] \phi  \;\mid\;  \nu Z. \phi \]
In $\nu Z. \phi$ 
every free occurrence of $Z$ in $\phi$ occurs positively, that is  within
the scope of an even number of negations. 
If a formula is written as $\phi(Z)$, it is to be understood that the
subsequent writing of $\phi(\psi)$ means $\phi$ with $\psi$ substituted for all
free occurrences of $Z$. 

The positivity requirement on the fixpoint operator is a syntactic
means of ensuring that $\phi(Z)$ denotes a functional monotonic in $Z$,
and so has unique minimal and maximal fixed points. It is usually  
more convenient to
introduce derived dual operators, and work in
positive form: $\phi_1 \por  \phi_2$ means $\pnot(\pnot\phi_1 \pand  \pnot\phi_2)$,  $\langle a\rangle\phi$ means $\pnot[a]\pnot\phi$
and $\mu Z.\phi(Z)$ means $\pnot\nu Z.\pnot\phi(\pnot Z)$. 
A formula is in \w{positive form} if it is
written with the derived operators so that $\pnot$ only occurs applied to
atomic propositions. It is in \w{positive normal form} if in addition
all bound variables are distinct. Any closed formula can be put into positive 
normal form. It is also useful to have derived propositional constants
$\prtt$ (for $P \por \neg P$) and $\pff$ (for $P \pand \neg P$). 

A modal mu-calculus \w{structure} $\T$ (over $\Prop, \Act$) is a labelled
transition system, namely a set $\S$ of states and a family of transition
relations $\by{a} \, \subseteq  \S \times \S$ for  $a \in \Act$,  
together with an interpretation $\V_\Prop\colon \Prop \to  2^{\S}$ for
the atomic propositions.  As usual we write $s \by{a} t$ for 
$(s,t) \in \,  \by{a}$. 

Given a structure $\T$ and an interpretation $\V \colon \Var \to 2^{\S}$ 
of the variables, the set $\dentv{\phi}$ of states satisfying a formula
$\phi$ is defined as follows:
\begin{eqnarray}
\dentv{P} &{}={}& \V_\Prop(P) \\
\dentv{Z}  &{}={}& \V(Z) \\
\dentv{\pnot\phi} &{}={}& \S - \dentv{\phi} \\
\dentv{\phi_1 \pand  \phi_2} &{}={}& \dentv{\phi_1} \cap  \dentv{\phi_2} \\
\dentv{[a]\phi} &{}={}& \eset{s \, | \, \forall t.
\mbox{if } s \by{a} t \mbox{ then }  t \in  \dentv{\phi} }\\
\dentv{\nu Z.\phi} &{}={}& \bigcup \big\{ S \subseteq  \S \, | \, 
 S \subseteq  \den{\phi}^\T_{\V[Z:=S]} \big\}
\end{eqnarray}
where $\V[Z:=S]$ is the valuation which maps $Z$ to $S$ and otherwise
agrees with $\V$. If we are working in positive normal form, we may
add definitions for the derived operators by duality (and for the propositional
constants). 
\begin{eqnarray}
\dentv{\phi_1 \por  \phi_2} &{}={}& \dentv{\phi_1} \cup  \dentv{\phi_2} \\
\dentv{\langle a\rangle\phi} &{}={}& \eset{s \, | \, \exists t.s \by{a} t \pand  t \in  \dentv{\phi} }\\
\dentv{\mu Z.\phi} &{}={}& \bigcap \big\{ S \subseteq  \S \, | \, 
 S \supseteq  \den{\phi}^\T_{\V[Z:=S]} \big\} \\
\dentv{\prtt} &{}={}& \S \\
\dentv{\pff}  &{}={}& \emptyset
\end{eqnarray}

If we take the usual lattice structure on $2^{\S}$, given by set inclusion, and 
if $f$ is a monotonic function then by the Knaster-Tarski theorem $f$ has fixed 
points, and indeed has a unique maximal and a unique minimal fixed point.
The maximal fixed point is the union of
\w{post-fixed points}, $\bigcup \eset{ S \subseteq \S \, | \, 
S \subseteq f(S)}$, and the minimal fixed point 
is the intersection of \w{pre-fixed points}, $\bigcap \eset{ S \subseteq \S \, | \,
f(S) \subseteq S }$. These determine the meanings of $\nu$ and $\mu$ in $\mu M$.
 
Moreover, the standard theory of fixpoints tells 
that if $f$ is a monotone function on a
lattice, we can construct its minimal fixed point  by applying
$f$ repeatedly on the least element  of the lattice
to form an increasing
chain,  whose limit is the least fixed point.
Similarly,  the maximal fixed point is constructed
by applying $f$ repeatedly on the largest element
to form a decreasing chain, whose limit is the maximal  fixed point.
The stages of these iterations  can be introduced syntactically as 
$\mu^{\alpha} Z. \phi$ and $\nu^{\alpha} Z. \phi$ for ordinals $\alpha$
whose meanings are  as follows when  $\lambda$ is a limit ordinal.

\begin{eqnarray}
\dentv{\mu^{0} Z. \phi} &{}={}& \emptyset \\
\dentv{\nu^{0} Z. \phi} &{}={}& \S \\
\dentv{\mu^{\beta + 1} Z. \phi}  &{}={}& \dentv{\phi(\mu^{\beta}Z.\phi)} \\
\dentv{\nu^{\beta + 1} Z. \phi}  &{}={}&\dentv{\phi(\nu^{\beta}Z.\phi)} \\
\dentv{\mu^{\lambda}Z.\phi} &{}={}& \bigcup_{\beta < \lambda} 
\dentv{\mu^{\beta} Z. \phi} \\
\dentv{\nu^{\lambda}Z.\phi} &{}={}& \bigcap_{\beta < \lambda} 
\dentv{\nu^{\beta} Z. \phi} \\
\end{eqnarray} 

\begin{defin} The formula $\phi$ of $\mu M$ is \w{valid} if for all
structures $\T$ and interpretations $\V$, $\dentv{\phi} = \S$.
The formula $\phi$ is \w{satisfiable} if there is a structure
$\T$ and an interpretation $\V$ such that $\dentv{\phi} \not= \emptyset$.
\end{defin}

\ni
As is standard $\models \phi$ indicates that $\phi$ is valid and 
$s \in \dentv{\phi}$ is written as  $s \models_{(\T,\V)} \phi$, dropping the
index $(\T,\V)$ wherever possible. 

The relationship between stages of iteration and
the fixpoints is formally described.

\begin{fact}
\label{fact1}   
\begin{enumerate}
\item $s \models \nu Z. \phi$ iff  $s \models \nu^{\alpha} Z. \phi$ for all 
ordinals $\alpha$.
\item $s \models \mu Z. \phi$ iff $s \models \mu^{\alpha} Z. \phi$ for some
ordinal $\alpha$.  
\end{enumerate}
\end{fact}
So for a minimal fixpoint formula 
$\mu Z.\phi$, if $s$ satisfies the 
fixpoint, it satisfies some iterate, say the
$\beta+1\,$th so that
$s \models \mu^{\beta+1} Z.\phi$. 
Now if we \w{unfold} this formula once, we get
$s \models  \phi(\mu^{\beta} Z.\phi)$. Therefore, 
the fact that $s$ satisfies the fixpoint
depends, via $\phi$, on the fact that other states in $\S$
satisfy the fixpoint
\emph{at smaller iterates than $s$ does}. So if one follows a chain
of dependencies, the chain terminates. 
Therefore, $\mu$ means
`finite looping'.
On the other hand, for a maximal fixpoint $\nu X.\phi$, there is no such
decreasing chain: $s \models  \nu Z.\phi$ iff 
$s \models \nu^{\beta} Z.\phi$ for every iterate
$\beta$ iff $s \models  \phi(\nu^{\beta} Z.\phi)$ for every iterate
$\beta$
iff $s \models  \phi(\nu  Z.\phi)$, and so we may loop
for ever. 

We impose a further syntactic constraint on formulas. In the following we write
$\sigma Z. \phi$ for $\mu Z. \phi$ or $\nu Z.\phi$ when we are indifferent to 
which  fixpoint.

\begin{defin}
The formula $\gamma$ of $\mu M$ is \w{guarded} if for any subformula
$\sigma Z. \phi$ of $\gamma$, every occurrence of $Z$ in $\phi$ is within 
the scope of a modal operator. 
\end{defin}

\ni
The following is standard; see \cite{Koz83,NiWa96,Wal00}.

\begin{fact} 
\label{guarded} Every formula of $\mu M$  is equivalent to a guarded
formula.
\end{fact}

\section{Proof Systems}
\label{secproof}

There has been a variety of proof systems for $\mu M$. Kozen presented
an equational deductive system which is equivalent to the  
Henkin  axiom system  of Figure~\ref{axiom} that extends 
the standard modal logic $K$ \cite{Koz83}: here $\phi \rightarrow \psi$
means $\pnot \phi \por \psi$. 
\begin{figure}
\[ \mbox{axioms and rules for minimal multi-modal logic K} \]
\[ \phi(\mu X.\phi(X)) \rightarrow  \mu X.\phi(X)\]
\[ \proofrule{\phi(\psi) \rightarrow \psi}{\mu X.\phi(X) \rightarrow \psi} \]
\caption{Kozen's axiomatisation of $\mu M$}
\label{axiom}
\end{figure}
There is an extra axiom for a  least fixed point that its ``unfolding''
implies it; and  Park's fixed point induction rule 
which says that $\mu$ is indeed the least pre-fixed point.
The duals of this axiom and rule for greatest fixed points are;
$\nu X. \phi(X) \rightarrow \phi(\mu X.\phi(X))$ and if
$\psi \rightarrow \phi(\psi)$ then $\psi \rightarrow
\nu X.\phi(X)$. Despite the naturalness of this axiomatisation, Kozen was
unable to show that it was complete in \cite{Koz83}. 
Instead, he proved it 
complete for a subset of $\mu M$, the 
aconjunctive fragment. Subsequently, he provided a complete infinitary
deductive system for the whole of $\mu M$  by adding the following 
infinitary rule \cite{Koz86}.

\[  \proofrule{\mu^{n} X. \phi(X)  \rightarrow \psi \mbox{ for all } n < \omega}
{\mu X.\phi(X) \rightarrow \psi} \]
Soundness of this rule depends on the \w{finite model property} which is
that  a formula
is satisfiable if, and only if,  it is satisfiable in a finite model. 
It is possible to devise an infinite structure (with infinite branching)
with  state
$s$ such that, for instance, $s \models \mu X. [a] X$ and 
$s \not\models \mu^{n} X. [a] X$ for all $n <
\omega$.  
 
Later Walukiewicz established that indeed  Kozen's axiomatisation 
in Figure~\ref{axiom} is 
complete for the whole language. The proof appeals to 
a normal form, \w{disjunctive normal form}, 
inspired by automata and semantic tableaux 
and also uses  (a slightly weakened version of)
aconjunctivity \cite{Wal00}. First, it is shown that 
every formula is \w{provably} equivalent to a guarded formula
(thereby strengthening Fact~\ref{guarded}).
For any  unsatisfiable weakly aconjunctive or disjunctive normal 
form formula $\phi$ there is a proof of $\pnot \phi$. 
Then the central  argument 
proceeds by induction on formulas  showing that every guarded
formula provably implies a semantically equivalent disjunctive normal form
formula. 
This unusual proof method for showing completeness
can be contrasted with the more  standard technique of building a model
out of consistent sets of formulas (which has remained elusive
for $\mu M$). 

Given a valid formula such as $\nu Z. \mu X. [a]Z \por \dia{a} X$
it is not so easy to provide a proof of it within Kozen's  axiom  system. 
This suggests that one may also seek natural deduction,  sequent or
tableau  style
proof  systems.  A  \w{goal directed} 
proof system is presented in Figure~\ref{fig2}.
\begin{figure}
\[  \Gamma, P, \pnot P  \ \ \ \ \ \ \  \Gamma, \prtt \]
\[ \proofrule{\Gamma, \phi \por \psi}{\Gamma, \phi, \psi} \ \ \ \ \ \
\proofrule{\Gamma, \phi \pand \psi}{\Gamma, \phi \ \ \ \ \ \ \ \ \ \Gamma, \psi}
\]
\[\proofrule{\Gamma, \dia{a} \Sigma, [a] \psi}
{\Sigma,\psi} \]
\[ \proofrule{\Gamma, \nu Z. \phi(Z)}{\Gamma, \phi(\nu Z. \phi(Z))} \ \ \ \ \ \
\proofrule{\Gamma, \mu Z. \phi(Z)}{\Gamma, \phi(\mu Z. \phi(Z))} \]
\caption{Goal directed proof rules} 
\label{fig2}
\end{figure}
A sequent of this proof system is a set of formulas understood
disjunctively; we assume $\Gamma, \Sigma, \ldots$ indicate a \w{set}
of formulas and $\Gamma, \phi, \psi$  is the set
$\Gamma \cup \eset{ \phi,\psi }$; clearly, $\Gamma, P, \neg P$ and $\Gamma,
\prtt$ are then valid. The rules remove $\por$ between formulas and branch
at an $\pand$. Some notation in the modal rule: $\dia{a} \Sigma$
is the set of formulas $\eset{ \dia{a}\phi \, | \, \phi \in \Sigma}$.
In its application the set $\Sigma$ can be empty.
Fixpoint formulas are unfolded. The idea is to build
a proof for a starting guarded formula $\gamma$ in positive normal form.
Such systems have been presented before. For instance,
in \cite{NiWa96} there is a dual system
for showing that a formula is unsatisfiable.
There are also systems, such as in \cite{DHL06,JKS08,Stu08},  
where the rules are
inverted. 

The main problem with the rules in Figure~\ref{fig2} is
that they lead to  infinite depth proof trees  as in
Figure~\ref{fig3}.  It is unclear when such a tree is in
fact a proof; for instance, there are such trees
for invalid formulas such as 
$\mu X. [a]X \por \dia{a} X$.  
\begin{figure}
\[ \tree{
\PTirule{$\nu Z. \mu X. [a] Z \por \dia{a} X$}
{\PTirule{$\mu X. [a] (\nu Z. \mu X. [a] Z \por \dia{a} X) \por \dia{a} X$}
{\PTirule{$[a](\nu Z. \mu X. [a] Z \por \dia{a}X) \por \dia{a}( 
\mu X. [a] (\nu Z. \mu X. [a] Z \por \dia{a} X) \por \dia{a} X )$}
{\PTirule{$[a](\nu Z. \mu X. [a] Z \por \dia{a}X),  \dia{a}( 
\mu X. [a] (\nu Z. \mu X. [a] Z \por \dia{a} X) \por \dia{a} X)$}
{\PTirule{$\nu Z. \mu X. [a] Z \por \dia{a} X,  \mu X. [a] (\nu Z. \mu X. 
[a] Z \por \dia{a} X) \por \dia{a} X$}
{\PTirule{$\mu X. [a] (\nu Z. \mu X. [a] Z \por \dia{a} X) \por \dia{a} X$}
{\PTaxiom{$\vdots \vdots$}}}}}}}}
\]
\caption{A never ending proof tree} 
\label{fig3}
\end{figure}
One solution is to replace  infinite depth  proofs with proofs of infinite
width by adopting a variant
of  Kozen's infinitary rule. 
In \cite{JKS08, Stu08} the authors add 
an infinitary rule as follows (again whose soundness
depends on the finite model property).  

\[ \proofrule{\Gamma, \nu Z. \phi(Z)}{\Gamma,\nu^{1} Z. \phi(Z) \ \ \
\ldots \ \ \ \Gamma, \nu^{n}Z. \phi(Z) \ \ \ \ldots} \]
 
\[ \proofrule{\Gamma, \nu^{1} Z. \phi(Z)}{\Gamma, \phi(\prtt)}
\ \ \ \ \ \ \ \
\proofrule{\Gamma, \nu^{n+1} Z. \phi(Z)}{\Gamma, \phi(\nu^{n} Z. \phi(Z))} 
\]
Every branch in a successful  proof tree thereby is  finite and finishes
at a sequent $\Gamma, \prtt$ or $\Gamma, P, \neg P$. For
instance, Figure~\ref{fig30} illustrates part of the proof tree
for $\nu Z. \mu X. [a] Z \por \dia{a} X$.
\begin{figure}
\[ Z^{i} = \nu^{i} Z. \mu X. [a] Z \por \dia{a} X  \ \mbox{ for $i > 0$} \]
\[ \tree{
\PTiiirule{$\nu Z. \mu X. [a] Z \por \dia{a} X$}
{\PTirule{$Z^{1}$}
{\PTirule{$\mu X. [a] \prtt \por \dia{a} X$}
{\PTirule{$[a] \prtt \por \dia{a}(\mu X. [a] \prtt  \por \dia{a} X)$}
{\PTirule{$[a] \prtt, \dia{a}(\mu X. [a] \prtt  \por \dia{a} X)$}
{\PTaxiom{$\prtt$}}}}}}
{\PTirule{$Z^{2}$}
{\PTirule{$\mu X. [a]Z^{1} \por \dia{a} X$}
{\PTirule{$[a] Z^{1} \por \dia{a}(\mu X. [a] Z^{1}  \por \dia{a} X)$}
{\PTirule{$[a] Z^{1}, \dia{a}(\mu X. [a] \prtt  \por \dia{a} X)$}
{\PTirule{$Z^{1}$}
{\PTaxiom{$\vdots$}}}}}}}
{\PTirule{$Z^{i+1}$}
{\PTirule{$\vdots$}
{\PTirule{$Z^{i}$}
{\PTaxiom{$\vdots$}}}}}} \]
\caption{An infinitely wide proof tree} 
\label{fig30}
\end{figure}

Alternatively, one can accept infinite depth proofs but find a finite 
way of generating or recognising  them. 
Extra criteria for deciding when an infinite tree labelled with sets
of formulas is indeed a proof are necessary. In particular, we need to guarantee
(see comments after Fact~\ref{fact1}) that in any infinite branch
a greatest    fixpoint 
formula is  unfolded infinitely often. 
In \cite{NiWa96} the authors  add the  extra 
mechanism of an   infinite  game that plays
over an  infinite tree. 
In \cite{DHL06} 
for linear time mu-calculus the extra mechanism is a nondeterministic parity
automaton that runs over the tree. 

What we shall do is to show that indeed there is a means for obtaining a 
finite proof using names. This mechanism was introduced in \cite{Ju09}
as a  tableau decision procedure for showing
satisfiability of $\mu M$ formulas. Here we reformulate it as a proof system
for showing when a formula is valid.

\section{Proof System with Names}
\label{secnames}
Our aim is now to build a proof system such that a formula  
has a finite proof tree
if, and only if, it is valid. The proof system  includes 
some auxiliary naming notation.
Assume a starting guarded closed formula $\gamma$ in positive
normal form. 

\begin{defin}
If in $\gamma$ the subformula $\sigma_{1} Z. \psi$ is
a proper subformula of $\sigma_{2} Y. \phi$ then $Y$ is \w{more outermost} 
than $Z$
(in $\gamma$).  Variable $X$ is a 
\w{variable} in $\gamma$ if $\sigma X.\psi$ 
is a subformula of $\gamma$ and it is a \w{$\nu$-variable} if $\sigma$
is $\nu$.
\end{defin} 

We assume a fixed linear ordering  $X_{1}, \ldots, X_{m}$
on all the distinct  variables in $\gamma$ such that if $X_i$
is more outermost than $X_j$ then $i < j$. 
For instance, in a linear ordering for variables in  
$(\nu Z. \mu X. [a] Z \por \dia{a} X) \pand \mu Y. [a] Y$ 
the  $\nu$-variable $Z$ must occur before $X$ whereas $Y$ can 
occur before or after it. 
For each $\nu$-variable $Z$ in $\gamma$ we assume a finite set 
$\eset{z_{1}, z_{2}, \ldots, z_{l}}$
of \w{names} for $Z$ where $l$ is the length of $\gamma$.

The proof system has sequents of the form $w \vdash \Gamma$ where
$w$ is a sequence of distinct names for $\nu$-variables and
each element of $\Gamma$ has the form  $\phi^{u}$ 
where $\phi$ is a formula (belonging to the closure of $\gamma$)  
and $u$ is a  subsequence of $w$. The initial sequent is
$\vdash \gamma$ with the empty sequence of names.
If $v = n_{1} \ldots n_{k}$ is a sequence of names then $v(i)$,
$1 \leq i \leq k$, is the element $n_{i}$. 

\begin{defin}
Assume  $X_{1}, \ldots, X_{m}$ is the fixed  linear ordering  of
variables  in $\gamma$  and $u, v, w$ are sequences of names
of these variables where  $u, v$ are subsequences of $w$.
\begin{enumerate}
\item We write $u <_{w} v$ if for some $j$, 
(1) $u(j)$ and $v(j)$ are names of the same variable
and $u(j)$ occurs before $v(j)$ in $w$, and (2) $u(i) = v(i)$
for all $i < j$.
\item  The sequence $u \restriction X_i$ is the subsequence of
$u$ that omits all names of the variables $X_{i+1}, \ldots, X_n$.
\item We write $u \sqsubset_{w} v$  if $u <_{w} v$ or there is a $\nu$-variable
$X_{i}$ such that $v \restriction X_{i}$ is a proper prefix of $u \restriction
X_{i}$.
\end{enumerate}
\end{defin}

The proof rules in Figure~\ref{fig4} 
are an elaboration of those in Figure~\ref{fig2}.
Again, sets of formulas are to be understood disjunctively; now formulas also
carry  sequences of names reflecting the history of 
unfoldings of greatest fixpoints.
\begin{figure}
\[ w \vdash  \Gamma, P^{u}, \pnot P^{v}  \ \ \ \ \ \ \ \ \ \  \ \  \ \ 
w \vdash \Gamma, \prtt^{u} \]
\[ \proofrule{w \vdash \Gamma, \phi \por \psi^{u}}{w \vdash \Gamma, \phi^{u}, 
\psi^{u}} \ \ \ \ \ \ \ \ \ \ \ \ \ 
\proofrule{w \vdash \Gamma, \phi \pand \psi^{u}}
{w \vdash \Gamma, \phi^{u} \ \ \ \ \ \ \ \ \ w \vdash \Gamma, \psi^{u}}
\]
\[\proofrule{w \vdash \Gamma, \dia{a} \Sigma, [a] \psi^u}
{w' \vdash \Sigma,\psi^u} \ \ \ \ \ \ \ \ \ \ \ \  \ \ \ 
\proofrule{w \vdash \Gamma, \mu Z. \phi(Z)^{u}}{w' \vdash \Gamma, 
\phi(\mu Z. \phi(Z))^{u \restriction Z}} \]
\[ \proofrule{w \vdash \Gamma, \nu Z. \phi(Z)^u}
{w'z_{i} \vdash \Gamma, \phi(\nu Z. \phi(Z))^{(u \restriction Z)z_{i}} } \
z_{i} \mbox{ is the first name for $Z$ not occurring  in $w$}
\]
\caption{Goal directed proof rules with names} 
\label{fig4}
\end{figure}
The $\por$ and  $\pand$ rules are similar to before; the names index
is passed to the components. In the modal rule
we assume that $\dia{a} \Sigma$
is the set of formulas $\eset{ \dia{a}\phi^{u} \, | \, \phi^{u} \in \Sigma}$;
in an  application $\Sigma$ can be empty. Some further notation:
$w'$ in the conclusion of the modal rule (and in other rules)
is the subsequence of names in  $w$ that still occur in $\Sigma$ and $u$;
names that occurred only in formulas in the premises $\Gamma$ are removed from
$w$. Fixpoint formulas are unfolded; names in $u$ that belong to variables that
are more innermost than $Z$ are removed from $u$ (and from $w$ if they do not
occur in $\Gamma$). In the case of a greatest  fixpoint a new name for $z$ is
also added to the name sequence (both in $w'$ and $u \restriction Z$). 
Importantly, there are also two key structural rules in Figure~\ref{fig5}.
\begin{figure}

\[ \mbox{Thin} \ \proofrule{w \vdash \Gamma, \phi^{u}, \phi^{v}}{w' \vdash \Gamma, \phi^{u}} \
u \sqsubset_{w} v \]

\[ \mbox{Reset$_z$} \  \proofrule{w \vdash \Gamma, \phi_{1}^{uzz_{1}u_{1}},
\ldots, \phi_{k}^{u z z_{k} u_{k}}}
{w' \vdash \Gamma,\phi_{1}^{uz}, \ldots, \phi_{k}^{uz}}\ \mbox{$z$ does not
occur in $\Gamma$} \]

\caption{Structural  proof rules} 
\label{fig5}
\end{figure}
If $\phi^{u}$ and $\phi^{v}$ both occur in  a sequent $w \vdash \Sigma$ then
either $u \sqsubset_{w} v$ or $v \sqsubset_{w} u$.  In the case of the 
rule Reset$_{z}$ the names $z, z_{1}, \ldots, z_{k}$
are names for the same variable $Z$  and $z_{i}$ could be the same as $z_{j}$.
When applying the proof rules of Figures~\ref{fig4} and \ref{fig5} we assume
that the structural rules have  priority over the logical rules.

\begin{defin}
\label{leaf}
A node $n$ of a tree labelled with the sequent
$w \vdash \Gamma$ is a \emph{leaf} if
there is a node $m$ above it, its \emph{companion}, 
labelled with the same sequent
$w \vdash \Gamma$; this leaf is \emph{successful} if between nodes
$n$ and $m$ there is an application  of the rule Reset$_{z}$ for some $z$
such that for any node $n'$ labelled with $w' \vdash \Sigma$ between
and including $n$ and $m$ the name $z$ occurs in $w'$.
\end{defin}

\begin{defin}
A \emph{proof
tree} for $\gamma$ is a tree where
\begin{enumerate}
\item the root is labelled $\vdash \gamma$,
\item any  other node is labelled with a  sequent that is the result of an 
application of a rule in  Figure~\ref{fig4} or \ref{fig5} to the sequent
at its parent node, 
\item each leaf is labelled with a sequent
that is  an instance of an  axiom in Figure~\ref{fig4}
or is successful according to  the repeat condition. 
\end{enumerate}
\end{defin}

\ni
A tree is not a proof if it has a leaf labelled with a
sequent of the form 
\[ w \vdash P_{1}^{u_{1}}, \ldots, P_{k}^{u_{k}}, \pnot Q_{1}^{v_{1}},
\ldots, \pnot Q_{l}^{v_{l}}, \dia{a_{1}} \Sigma_{1}, \ldots, \dia{a_{m}} 
\Sigma_{m} \]
where $Q_{j} \not= P_{i}$ for all $i,j$ or has a leaf $n$ that is a repeat
because of its companion  $m$  and for every application of a rule 
Reset$_{z}$ between $m$ and $n$ there is a node $n'$ between (and including) 
$n$ and $m$ labelled $w \vdash \Sigma$ such that $z$ does not occur in $w$. 
Given a formula $\gamma$ there are at most $2^{ | \gamma | }$ different subsets
of subformulas of $\gamma$ where $| \gamma |$ is the size of $\gamma$.
The number of greatest fixpoints in $\gamma$ is also bounded by $| \gamma |$.
The number of different possible sequents derivable from $\vdash \gamma$
is bounded by $2^{O(|\gamma|^{2} |log(\gamma)|)}$, see \cite{Ju09}, which is therefore
also a bound on the depth of a tree.  Moreover, the width of  a tree
is bounded by $2$.  The only rule that allows choice is the modal rule;
the number of choices is again bounded by $|\gamma|$. 
Therefore, the  number of possible trees with root $\vdash \gamma$ is 
bounded in terms of $| \gamma |$. 

\begin{fact}
\label{finitefact}
For any closed guarded $\gamma$ there are only boundedly  many 
trees
for $\gamma$ and each such tree has boundedly many nodes (where the bounds
are functions  of $|\gamma|$).
\end{fact}
 
\begin{figure}
\[ Z = \nu Z. \mu X. [a] Z \por \dia{a} X  \ \ \ \ \  \ \    
X = \mu X. [a] Z \por \dia{a} X
\]

\[ \tree{
\PTirule{$\vdash Z$}
{\PTirule{$z_{1} \vdash X^{z_1}$}
{\PTirule{$z_{1} \vdash ([a] Z  \por \dia{a} X)^{z_{1}}$}
{\PTirule{$z_{1} \vdash [a] Z^{z_{1}},  
\dia{a} X^{z_{1}}$}
{\PTirule{$z_{1} \vdash Z^{z_{1}},  X^{z_{1}}$}
{\aPTirule{Thin}{$z_{1} z_{2} \vdash X^{z_{1} z_{2}}, X^{z_{1}}$}
{\aPTirule{Reset$_{z_{1}}$}{$z_{1} z_{2} \vdash X^{z_{1} z_{2}}$}
{\PTaxiom{$z_{1} \vdash X^{z_{1}}$}}}}}}}}}
\]
\caption{A proof tree} 
\label{fig6}
\end{figure}
In Figure~\ref{fig6} there is a proof tree for the valid formula
$\nu Z. \mu X. [a] Z \por \dia{a} X$ where we employ the abbreviations
that $Z$ is this formula and $X$ is it's subformula 
$\mu X. [a] z \por \dia{a} X$.
It is a  proof tree because of the repeat sequent 
$z_{1} \vdash X^{z_{1}}$ with an application of Repeat$_{z_{1}}$ inbetween
where $z_{1}$ is a name that occurs in each sequent throughout. 
The proof tree for a  more complex  valid formula $X \por Z$ is 
illustrated in 
\begin{figure}
\[ \begin{array}{ll} 
X = \nu X. \dia{a}X \pand Y \ \ \ \ \  \ \ &   
Z = \nu Z. [a] Z \por W \\
Y = \mu Y. \dia{a} Y \por P  & W = \mu W. [a] W \por \neg P
\end{array} \]
\[ \tree{
{\PTirule{$\vdash X, Z$}
{\PTirule{$x_{1} \vdash (\dia{a} X \pand Y)^{x_1}, Z$}
{\PTirule{$x_{1} z_{1} \vdash (\dia{a} X \pand Y)^{x_1},
([a] Z  \por W)^{z_{1}}$}
{\PTiirule{$x_{1} z_{1} \vdash (\dia{a} X \pand Y)^{x_1}, [a] Z^{z_{1}},  
W^{z_{1}}$}
{\PTaxiom{T1}}
{\PTirule{$x_{1} z_{1} \vdash Y^{x_{1}}, [a] Z^{z_{1}}, W^{z_{1}}$}
{\PTirule{$x_{1} z_{1} \vdash (\dia{a} Y \por P)^{x_{1}}, [a] Z^{z_{1}}, W^{z_{1}}$}
{\PTirule{$x_{1} z_{1} \vdash \dia{a} Y^{x_{1}}, P^{x_{1}}, [a] Z^{z_{1}}, W^{z_{1}}$}
{\PTiirule{$x_{1} z_{1} \vdash \dia{a} Y^{x_{1}}, P^{x_{1}}, [a] Z^{z_{1}}, 
([a]W \pand \neg P)^{z_{1}}$}
{\PTirule{$x_{1} z_{1} \vdash  \dia{a} Y^{x_{1}}, P^{x_{1}}, [a] Z^{z_{1}}, 
[a]W^{z_{1}}$}
{\PTirule{$x_{1} z_{1} \vdash Y^{x_{1}}, Z^{z_{1}}$}
{\PTirule{$x_{1} z_{1} z_{2} \vdash Y^{x_{1}}, ([a]Z \por W)^{z_{1} z_{2}}$}
{\aPTirule{Reset$_{z_{1}}$}{$x_{1} z_{1} z_{2} \vdash Y^{x_{1}}, [a]Z^{z_{1}z_{2}}, W^{z_{1} z_{2}}$}
{\PTaxiom{$x_{1} z_{1} \vdash Y^{x_{1}}, [a]Z^{z_{1}},  W^{z_{1}}$}}}}}}
{\PTaxiom{$x_{1} z_{1} \vdash \dia{a}Y^{x_{1}}, P^{x_{1}}, [a] Z^{z_{1}},  
\neg P^{z_{1}}$} }}}}}}}}}} \]

\[ \tree{
{\PTirule{T1}
{\PTirule{$x_{1} z_{1} \vdash \dia{a} X^{x_{1}}, [a] Z^{z_{1}},  W^{z_{1}}$}
{\PTiirule{$x_{1} z_{1} \vdash \dia{a} X^{x_{1}}, [a] Z^{z_{1}},  
([a] W \pand \neg P)^{z_{1}}$}
{\PTirule{$x_{1} z_{1} \vdash \dia{a} X^{x_{1}}, [a] Z^{z_{1}}, [a]  W^{z_{1}}$}
{\PTirule{$x_{1} z_{1} \vdash X^{x_{1}}, Z^{z_{1}}$}
{\aPTirule{Reset$_{x_{1}}$}{$x_{1} z_{1}x_{2} \vdash (\dia{a}X \pand Y)^{x_{1} x_{2}}, Z^{z_{1}}$}
{\PTirule{$x_{1} z_{1} \vdash (\dia{a} X \pand Y)^{x_{1}}, 
Z^{z_{1}}$}
{\aPTirule{Reset$_{z_{1}}$}{$x_{1} z_{1} z_{2}\vdash (\dia{a} X \pand Y)^{x_{1}}, 
([a] Z \por W)^{z_{1} z_{2}}$}
{\PTaxiom{$x_{1} z_{1} \vdash (\dia{a} X \pand Y)^{x_{1}}, 
([a] Z \por W)^{z_{1}}$}}}}}}}
{\PTirule{$x_{1} z_{1} \vdash \dia{a} X^{x_{1}}, [a] Z^{z_{1}}, \neg P^{z_{1}}$}
{\PTirule{$x_{1} z_{1} \vdash X^{x_{1}}, Z^{z_{1}}$}
{\aPTirule{Reset$_{x_{1}}$}{$x_{1} z_{1}x_{2} \vdash (\dia{a}X \pand Y)^{x_{1} x_{2}}, Z^{z_{1}}$}
{\PTirule{$x_{1} z_{1} \vdash (\dia{a} X \pand Y)^{x_{1}}, 
Z^{z_{1}}$}
{\aPTirule{Reset$_{z_{1}}$}{$x_{1} z_{1} z_{2}\vdash (\dia{a} X \pand Y)^{x_{1}}, 
([a] Z \por W)^{z_{1} z_{2}}$}
{\PTaxiom{$x_{1} z_{1} \vdash (\dia{a} X \pand Y)^{x_{1}}, 
([a] Z \por W)^{z_{1}}$}}}}}}}}}}}
\]
\caption{A proof tree} 
\label{fig7}
\end{figure}
Figure~\ref{fig7}. We encourage the reader to check that indeed
it is a  proof tree.

At the  cost of increasing the size of trees, 
we can add further  conditions on  when a node  counts
as a leaf in   Definition~\ref{leaf}: for instance, an extra
requirement is that  its sequent
is  the result of an 
application of the modal rule. 

\begin{theorem}
For any closed guarded $\gamma$, $\models \gamma$ iff there is a proof tree
for  $\gamma$. 
\end{theorem}
\begin{proof}
Assume $\models \gamma$ but there is not a proof tree for $\gamma$.
We show that we can build a countermodel to $\gamma$; a structure
$\T$ and  a state $s$ of $\T$ such that
$s \not\models \gamma$. 
Given a sequent $w \vdash \Gamma$ it is valid if 
$\models \bigvee \eset{ \phi \, | \, \exists u. \phi^{u} \in \Gamma}$.
The initial sequent $\vdash \gamma$ is valid.
We now build a tree using the proof rules where each node
is labelled with a valid sequent (or, as we shall see, a countermodel)
and except for the root node
is the result of an application of a  proof rule.
Assume we have built part of the tree and consider a current leaf
labelled with a valid sequent; if it is not an axiom or a repeat
then the tree can be extended with further valid sequents. This is clear
if we can apply a structural rule of Figure~\ref{fig5} which has priority
and it is also clear for $\pand$, $\por$ and the fixpoint
rules of Figure~\ref{fig4};
in all these cases if the premise sequent is valid 
then so are the conclusion sequents. 
We next come to the modal rule. We assume it is only applied if no other 
rule applies. Then a leaf of the current tree is labelled with a valid
sequent of the form
\[ (*) \ \ w \vdash P_{1}^{u_{1}}, \ldots, P_{k}^{u_{k}}, \pnot Q_{1}^{v_{1}},
\ldots, \pnot Q_{l}^{v_{l}}, \dia{a_{1}} \Sigma_{1}, \ldots, \dia{a_{m}} 
\Sigma_{m}, [b_{1}]\psi_{1}^{w_{1}}, \ldots, [b_{p}] \psi_{p}^{w_{p}} \]
where each $\Sigma_{i}$ is nonempty, $a_{i} \not= a_{j}$ when $i \not= j$
and we assume it is not an axiom, 
so $P_{i} \not= Q_{j}$ for all $i,j$. A  possible conclusion 
of an application of  the modal rule has the  form $w' \vdash 
\Sigma_{i}, \psi_{j}^{w_{j}}$ when  $a_{i} = b_{j}$ or $w' \vdash \psi_{j}^{w_{j}}$
when $b_{j}$ is different from each $a_{i}$. 
With our tree we allow \w{all} such possible applications.
For each such application if the sequent is not valid we let the node be 
a leaf and we associate a countermodel to it: that is, a structure
$\T_{ij}$ and a state $s_{ij}$ such that
$s_{ij} \not\models \bigvee \eset{ \phi \, | \exists u. \phi^{u}
\in \Sigma_{i}} \por \psi_{j}$ or a structure $\T_{j}$ and  a state $s_{j}$
such that $s_{j} \not\models \psi_{j}$. If all
possible applications of the rule are invalid, including the case
when $p =0$ in $(*)$, then we obtain a contradiction
by  constructing  a countermodel to the valid premise
$(*)$ as follows. For $\T$ we take the disjoint
union of each  $\T_{ij}$ and of each $\T_{j}$ together with a new state $s$.
For each $a_{i}$ such that $\neg \exists b_{j}. a_{i} = b_{j}$ assume there
is not  a transition of the form $s \by{a_{i}} s'$. Otherwise,
we let $s \by{a_{i}} s_{ij}$ of $\T_{ij}$ and $s \by{b_{j}} s_{j}$
of $\T_{j}$. 
Finally, we assume $s \not\in \V_{\Prop}(P_{i})$ and $s \in \V_{\Prop}(Q_{j})$
for each $i,j$.  Clearly, by construction, $s$ fails to satisfy each  formula
in $(*)$. Any node of the tree
labelled with a sequent of the form  $(*)$ is called a \emph{modal}
node. Therefore, there is at least one child node labelled with a valid
sequent of a modal node. For each such child we continue to extend the tree.
The tree building  eventually stops when nodes are leaves either because
they are children of a modal node labelled with an invalid sequent or
nodes 
labelled with an axiom or a repeat node.  In the last case we assume that we
restrict  repeat nodes to be children  of modal nodes.  
All nodes of the tree except for some successors of modal
nodes are labelled with valid sequents. 
However, by assumption there is not a proof tree for $\gamma$.
We now prune the tree. Starting top down, at any node where $\pand$ is applied
we choose one of the successor nodes which fails to produce a
proof tree; we discard the subtree
of  the other successor.
The result is a finite tree
where the  only branching is at modal nodes.
All leaves are either unsuccessful repeats or children of modal nodes
labelled with invalid sequents (and with associated countermodels). 
From this tree we build a countermodel
to $\gamma$. We identify as states any region of the tree starting at the
root or at a child of a modal node labelled with a valid sequent down to,
and including, the next modal node. In the case of a leaf that is  a 
repeat we assume that there is a backward edge to its companion node
above. If a state $s$ finishes at
the modal node labelled with the sequent $(*)$ then 
for each $a_{i}$ such that $\neg \exists b_{j}. a_{i} = b_{j}$ assume there
is not  a transition of the form $s \by{a_{i}} s'$. Otherwise, for each
child that is labelled with an invalid sequent 
we let $s \by{a_{i}} s_{ij}$ of the countermodel
$\T_{ij}$ or  $s \by{b_{j}} s_{j}$
of the countermodel $\T_{j}$. For any child labelled 
with valid sequent $w' \vdash 
\Sigma_{i}, \psi_{j}^{w_{j}}$ when  $a_{i} = b_{j}$ whose associated state is
$s'$ we assume a transition $s \by{a_{i}} s'$
or any child $w' \vdash \psi_{j}^{w_{j}}$ whose associated state is $s'$
we assume a transition $s \by{b_{j}} s'$: the associated state of a repeating
leaf is that of its companion (the target of the backedge).  
Finally, we assume $s \not\in \V_{\Prop}(P_{i})$ and $s \in \V_{\Prop}(Q_{j})$
for each $i,j$. We say that $\phi \in s$ if $\exists u. \phi^{u}$ belongs to 
some sequent in the region associated with $s$. The proof is completed by
showing that if $\phi \in s$ then  in the countermodel $s \not\models \phi$. 
Assume to the contrary that for some $s$ and $\phi$, $\phi \in s$ and
$s \models \phi$.
Clearly, then $\phi$ is not a literal, an atomic formula or the negation of an 
atomic formula. For a formula $\phi \in s$ we can follow it through the tree,
passing between states and jumping from a leaf to its companion. If $\phi_{1}
\pand \phi_{2} \in s$ then by construction $\phi_{1} \in s$ or $\phi_{2}
\in s$. If $\phi_{1} \por \phi_{2} \in s$ then we can choose
between  $\phi_{1} \in s$ and $\phi_{2} \in s$. If $\dia{a} \phi \in s$
then we look at the modal node associated with $s$: if there is not
a $t$ such tht $s \by{a} t$ or only countermodels under $a$-transitions
to $\phi$ then $s \not\models \dia{a} \phi$. Otherwise, we can choose a
$t$ such that $s \by{a} t$ and $\phi \in t$. Similarly, for $[b] \psi \in s$.
If $\sigma Z. \phi \in s$ then $\phi(\sigma Z. \phi) \in s$. 
Therefore, if we follow $\phi \in s$ for $s \models \phi$
we obtain a finite or infinite
sequence $\phi_{1} \in s_{1}, \phi_{2} \in s_{2}, \ldots, \phi_{n} \in s_{n}$
where $\phi_{1} = \phi$, $s_{1} = s$, there is a state transition
when $\phi$ is a modal formula and for all $i$, $s_{i} \models \phi_{i}$.
Clearly, the sequence cannot be finite  ending at a literal or a modal formula.
So, the sequence must be infinite. We show that the outermost fixpoint unfolded
infinitely often is  a least fixpoint which is a contradiction by
Fact~\ref{fact1}. For suppose it is a greatest fixpoint $\nu Z. \psi$:
then the sequence of formulas  must have a subsequence
of the form $\ldots, \nu Z. \psi^{u},
\psi(\nu Z. \psi)^{u'z}, \ldots, 
\nu Z. \psi^{u'zu_{1}}, \psi(\nu Z. \psi)^{u'z z_{i}}, \ldots, \nu Z. \psi^{u'zu_{2}}$
where Reset$_{z}$ is applied and $z$ is defined throughout: that is, 
the sequence  must pass through  a successful
repeat.

For soundness, assume that there is a proof tree
for $\gamma$ but $\not\models \gamma$. Therefore, there is
a proof tree with root labelled $ \vdash \gamma$ all of whose leaves
are either labelled with axioms or are successful repeats. 
A sequent $w \vdash \Gamma$ is \emph{not}  valid if $\not\models \bigvee
\eset{ \phi \, | \, \exists u. \phi^{u} \in \Gamma}$. 
First, if the premise of an application of a rule is not valid then
so is a conclusion.  This is clear for the structural rules,  for the
$\por$ rule and the fixpoint rules. In the case of $\pand$, 
if the premise sequent is not valid then one of the successor
sequents is not valid. In the case of the modal rule, if
$\models \bigvee \Sigma \por \Psi$ then by standard modal reasoning
$\models \phi \vee \dia{a} \Sigma \vee [a] \psi$; so, if the premise sequent
is not valid then neither is the conclusion in an application of the modal rule.
Next we refine the argument by adding ordinal information. If $\not\models
\nu Z. \phi$ then using Fact~\ref{fact1} there is a least ordinal $\alpha$,
a countermodel $\T$ and  a state $s$ of $\T$ such that
$s \not\models \nu^{\alpha} Z.\phi$. To do this, we slightly change the rules
(as in fact used in Figures~\ref{fig6} and \ref{fig7}) by letting variables 
abbreviate the fixpoint subformulas of $\gamma$.
\[ \proofrule{w \vdash \Gamma, \sigma Z. \phi(Z)^{u}}{w \vdash \Gamma, Z^{u}}
 \ \ \ \ \ \ 
\proofrule{w \vdash \Gamma, Z^{u}}{w' \vdash \Gamma, \phi(Z)^{u \restriction Z}}
\ Z \mbox{ is } \mu Z. \phi(Z) \]
\[ \proofrule{w \vdash \Gamma, Z^u}
{w'z_{i} \vdash \Gamma, \phi(Z)^{(u \restriction Z)z_{i}} } \
z \mbox{ is } \nu Z. \phi(Z) \mbox{ and }
z_{i} \mbox{ is the first name for $Z$ not occurring  in $w$}
\]
So, formulas can contain variables. 
We  associate ordinals
with sequents by adding ordinals to names. Assume an invalid
sequent $w \vdash \Gamma$ where $w = n_{1}, \ldots, n_{k}$.
We extend $w$ to pairs $(n_{1}, \alpha_{1}), \ldots, (n_{k},\alpha_{k})$
where each $\alpha_{i}$ is an ordinal: if $\phi^{u} \in \Gamma$
and $u$ contains a name for $Z$ then the meaning of $Z$
in $\phi^{u}$ is $\nu^{\alpha_{i}} Z. \psi$ when $Z$ is $\nu Z. \psi$
and  where $z_{i}$ is the 
last name for $z$ in $u$. We assume that the invalid sequent
$w \models \Gamma$ remains invalid when greatest
fixpoint subformulas are so interpreted.
We maintain the following invariant in an ordinal sequence:
if $w = (n_{1}, \alpha_{1}), \ldots, (n_{k},\alpha_{k})$, $i < j$ and 
$n_{i}, n_{j}$ name the same variable $Z$ such that there is a formula
$\phi^{u}$ such that $n_{i}, n_{j}$ both occur in $u$ then 
$\alpha_{i} > \alpha_{j}$. Moreover, we assume lexicographic ordering on 
ordinal sequences: if $w = (n_{1}, \alpha_{1}), \ldots, (n_{k},\alpha_{k})$
and $w' = (n_{1}, \beta_{1}), \ldots, (n_{k},\beta_{k})$ then
$w < w'$ if for some $j$, $\alpha_{j} < \beta_{j}$ and for 
all $i < j$, $ \alpha_{i} = \beta_{i}$.  We are interested
in a least ordinal interpretation which makes  $w \vdash \Gamma$ 
invalid. Moreover, if a proof rule is applied to such a sequent
then a conclusion is invalid under the ordinal interpretation;
we minimise the ordinal sequence which makes the conclusion invalid with 
respect to the lexicographical ordering.  
This is clear for the $\por$, Thin, $\pand$, modal, $\sigma Z$
and least fixpoint variable $Z$ 
(where we lose ordinals for any inner $X$ such that $Z > X$) rules.  
In the case of the maximal fixpoint variable  rule with premise 
$w \vdash \Gamma, Z^{u}$ if there is no name for $Z$ in $u$ then we know that
there is a least $\alpha$ such that 
$w' (z_{i}, \alpha)  \vdash \Gamma,  \phi(Z)^{(u \restriction  Z)z_{i}}$
is invalid where $z_{i}$ is a new name for $Z$. Otherwise,  there is a name
for $Z$ in $u$; suppose the last one is $z_{j}$ with ordinal $\alpha_{j}$.
Since the fixpoint is unfolded we know that we can decrease the meaning
of $Z^{u}$ by at least one; so for the invalid   conclusion 
$w' (z_{i}, \alpha)  \vdash \Gamma,  \phi(Z)^{(u \restriction  Z)z_{i}}$
$\alpha < \alpha_{j}$.  Finally, we turn to the Reset$_{z}$ rule with
premise $w \vdash \Gamma, \phi_{1}^{uzz_{1}u_{1}},
\ldots, \phi_{k}^{u z z_{k} u_{k}}$ where $z$ does not occur in $\Gamma$
and $z, z_{1}, \ldots, z_{k}$ name the same variable.
In  $w$ we have $(z, \alpha)$ and later $(z_{1}, \alpha_{1}), \ldots,
(z_{k}, \alpha_{k})$ (in any order). By the invariant property
it follows that $\alpha > \alpha_{i}$ for each $i$ and that 
$Z$ of $\phi_{j}$ has meaning $\nu^{\beta_{j}} Z. \phi$ for 
$\beta_{j} \leq \alpha_{j}$
(as $u_{j}$ may contain further names for $Z$). Let $\beta = 
\mathrm{min}\eset{\alpha_{1}, \ldots, \alpha_{k}}$. Clearly, we can
replace $(z,\alpha)$ in $w$ with $(z, \beta)$, remove all the names
$z_{i} u_{i}$ such that $w' \vdash \Gamma,\phi_{1}^{uz}, \ldots, \phi_{k}^{uz}$
is invalid.  Given a proof tree for $\gamma$ we now follow a branch
of invalid sequents down the tree minimising their ordinal
interpretations of variables. Clearly, we cannot reach a leaf $w \vdash \Gamma, P^{u},
\neg P^{v}$ or $w \vdash \Gamma, \prtt^{u}$ as these sequents are valid.
Moreover, we cannot reach a successful repeat $w \vdash \Gamma$ with
an application of Reset$_{z}$ in between when  $z$ is in each sequent 
throughout.
Consider the companion node with ordinal interpretation 
$w = (n_{1}, \alpha_{1}), \ldots, (n_{k}, \alpha_{k})$
and the leaf node with interpretation $w' = (n_{1}, \beta_{1}), \ldots, 
(n_{k}, \beta_{k})$: it follows that $w ' < w$ as at least the entry for $z$
was reduced by the Reset$_{z}$ rule which is a contradiction. 
\end{proof}

\section{Conclusion}

We have presented a sound and complete proof system  for checking validity
of modal mu-calculus formulas. However, it relies on auxiliary notation for
names that keep track of  unfoldings of greatest fixpoints. 

We  tried,  but failed,  to
see if this  method is able to underpin a different proof of completeness
of Kozen's axiomatisation  than  Walukiewicz's 
proof by induction.  

An alternative framework for deciding satisfiability and validity for
$\mu M$ is automata-theoretic  \cite{StE89}. Using two way automata
there is also a decision procedure for  satisfiability and validity
of formulas when past  
modal operators are  included
\cite{V98}. 
Neither a sound and complete 
axiom system nor a sound and complete tableau  proof system
have been developed for this extended fixpoint  logic (which fails 
the finite model property). 
 
\bibliographystyle{eptcs}

\end{document}